\documentclass[a4paper, intlimits, 10pt]{amsart}

\usepackage[T1]{fontenc}
\usepackage{float}

\usepackage{times}
\usepackage[english]{babel}
\usepackage[T1]{fontenc}
\usepackage{enumitem}
\usepackage[short]{optidef}

\usepackage{graphicx}
\usepackage{amsmath,bm}
\usepackage{caption}
\usepackage{subcaption}
\usepackage{booktabs}
\usepackage{tabularx}
\usepackage{color}
\usepackage{eurosym}
\usepackage[numbers,sort&compress]{natbib}

\usepackage{amsthm, amssymb, amsfonts}
\usepackage{amsmath}

\usepackage{accents}

\newtheorem{theorem}{Theorem}[section]

\newtheorem{proposition}[theorem]{Proposition}

\newtheorem{example}[theorem]{Example}

\newtheorem{remark}[theorem]{Remark}
\numberwithin{equation}{section}

\begin{document}

\title{Dual Representation of Expectile based expected shortfall and Its Properties}

\date{\today}

\author{Samuel Drapeau}
\thanks{National Science Foundation of China, Grants Numbers: 11971310 and 11671257; Grant ``Assessment of Risk and Uncertainty in Finance'' number AF0710020 from Shanghai Jiao Tong University; are gratefully acknowledged.}
\address{School of Mathematical Sciences \& Shanghai Advanced Institute for Finance (CAFR)\newline Shanghai Jiao Tong University, Shanghai, China}
\email{sdrapeau@saif.sjtu.edu.cn}
\urladdr{http://www.samuel-drapeau.info}

\author{Mekonnen Tadese}
\address{School of Mathematical Sciences\newline Shanghai Jiao Tong University, Shanghai, China}
\email{mekonnenta@sjtu.edu.cn}

\begin{abstract}
    The expectile can be considered as a generalization of quantile.
    While expected shortfall is a quantile based risk measure, we study its counterpart -- the expectile based expected shortfall -- where expectile takes the place of quantile.
    We provide its dual representation in terms of Bochner integral.
    Among other properties, we show that it is bounded from below in terms of convex combinations of expected shortfalls, and also from above by the smallest law invariant, coherent and comonotonic risk measure, for which we give the explicit formulation of the corresponding distortion function.
    As a benchmark to the industry standard expected shortfall we further provide its comparative asymptotic behavior in terms of extreme value distributions.
    Based on these results, we finally compute explicitly the expectile based expected shortfall for some selected class of distributions.

    \vspace{5pt}

\end{abstract}

\maketitle
\section{Introduction}

Many risk measures proposed for the quantification of financial risks are given either directly or indirectly in terms of quantile of the distribution of the loss profile.
This includes the value at risk $q_\alpha(L) :=\inf \{ m\colon P[ L>m ]\geq \alpha \}$ and the derivation of is such as the tail conditional expectation and the expected shortfall
\begin{equation*}
    TCE_\alpha(L) := E\left[ L|L>q_{\alpha}(L) \right] \quad \text{and} \quad
    ES_{\alpha}(L) := \frac{1}{\alpha}\int_{0}^{\alpha} q_{u}(L) du
\end{equation*}
respectively.
While the expected shortfall is coherent in the sense of \citet{artzner1999}, the value at risk and tail conditional expectation are not sub-additive and hence not coherent, see \citep{tasche2002, foellmer2016}.
For diversification purposes, the expected shortfall is therefore preferred to these two quantile based risk measures.
However, the expected shortfall is not elicitable, which has recently been discussed as a useful property from a backtesting viewpoint, see \citet{gneiting2011, ziegel2014, emmer2015, chen2018}.
In terms of elicitablity and coherency, the expectile $e_\alpha$ which was first introduced by \citet{newey1987} and defined as the unique solution of
\begin{equation*}
    (1-\alpha)E\left[ \left(L-e_{\alpha}(L)\right)^+ \right] = \alpha E\left[ \left( L-e_{\alpha}(L) \right)^- \right]
\end{equation*}
is the only alternative coherent law invariant risk measure which is elicitable as shown by \citet{weber2006}, \citet{ziegel2014} and  \citet{bellini2015b}.

As discussed by \citet{bellini2014}, the expectile can be seen as a generalization of quantile.
We therefore revisit the former quantile based risk measure by considering the expectile instead of quantile, namely
\begin{equation*}
    tce_\alpha (L)  := E\left[ L|L>e_\alpha(L) \right] \quad \text{and}\quad es_{\alpha}(L)  := \frac{1}{\alpha}\int_{0}^{\alpha}e_{u}(L)du 
\end{equation*}
thereafter referred to as expectile based tail conditional expectation and expectile based expected shortfall, respectively.

The notion of expectile based tail conditional expectation $tce_\alpha$ and expected shortfall $es_\alpha$ is relatively new.
Expectile based tail conditional expectation was first introduced in \citet{tylor2008} for the estimation of the expected shortfall from expectile for loss profile with continuous distribution.
Though it is positive homogeneous and cash invariant, it is however not monotone and sub-additive in general, see \citet{daouia2019}.
For this reason, \citet{daouia2019} criticized the estimation $ES_\alpha$ from $tce_\alpha$ and propose $es_\alpha$ which is coherent.
They further showed that for Fr\'echet type of extreme value distribution, $tce_\alpha$ and $es_\alpha$ are asymptotically equivalent.

In this paper we systematically study the properties of these expectile based risk measures under the light of recent results obtained in \citep{mekonnen2019}.
Since the expectile based expected shortfall $es_{\alpha}$ turns out to be a coherent risk measure, we provide its dual set in terms of Bochner integral
\begin{equation*}
    \mathcal{Q}_{es}=\left\{Y\in \mathcal{Q}\colon Y=\frac{1}{\alpha}\int_0^\alpha \mathbf{Y}(u)du \quad\text{for some } \mathbf{Y}\in \mathcal{Y} \text{ with } \int_0^\alpha ||\mathbf{Y}(u)||_\infty du <\infty \right\},
\end{equation*}
where $Q$ is the set of probability densities in $L^\infty$, and $\mathcal{Y}$ is the set of strongly measurable functions $\mathbf{Y}:(0,\alpha] \to L^\infty$ such that $\mathbf{Y}(u)$ is in the dual set of $e_{u}$ for almost every $u$.
We further bound $es_{\alpha}$ from below in terms of combinations of expected shortfall, that is
\begin{equation*}
    \sup_{0<\beta<1}\left\{ \left( 1-\gamma_\beta \right)ES_\beta(L)+\gamma_\beta E[L] \right\}\leq es_\alpha(L)
\end{equation*}
where $\gamma_\beta$ has an explicit expression given by relation \eqref{eq:gamma}.
Though $es_\alpha$ is not comonotonic, in the sense of \citet{delbaen2013} we provide the smallest comonotonic risk measure dominating $es_\alpha$, that is
\begin{equation*}
    es_{\alpha}(L) \leq R_{\varphi}(L):= \int_{0}^{1}\varphi^\prime_+(t) q_{t}(L) dt 
\end{equation*}
where the concave distortion function $\varphi$ is explicitly given by Relation \eqref{eq:primitive}.
To compare the value of $es_\alpha$ with respect to the industry standard expected shortfall and value at risk, we provide their asymptotic relative behavior for each extreme value distribution type -- Fr\'echet, Weibull and Gumbel.
Finally, based on the present result we provide explicit expression for $es_\alpha$ for several classical distributions.

The rest of the paper is organized as follows.
In Section 2, we introduce some basic notations, and definitions of the quantile and expectile based risk measures.
In Section 3, we address the dual representation of expectile based expected shortfall.
In Section 4, we provide properties, bounds and asymptotic results for the expectile based expected shortfall.
Section 5 illustrates those results with examples for some loss profiles with known distribution.

\section{Basic Definition and Preliminaries}\label{sec:sec02}
Let $(\Omega, \mathcal{F},P)$ be an atomless probability space.
Throughout, $L^1$ and $L^\infty$ denote the set of integrable and essentially bounded random variables identified in the $P$ almost sure sense, respectively.
For each $L$ in $L^1$, $F_L$ represents its cumulative distribution.
We also denotes by $\mathcal{Q}$, the set of densities in $L^\infty$ for probability measures that are absolutely continuous with respect to $P$, that is,
\begin{equation*}
    \mathcal{Q}=\left\{Y\in L^\infty\colon Y\geq 0 \text{ and } E[Y]=1\right\}.
\end{equation*}
We say $R\colon L^1 \to (-\infty,\infty]$ is a \emph{coherent risk measure}, if $R$ is 
\begin{itemize}
    \item \emph{Monotone}: $R(L_1)\leq R(L_2)$ whenever $L_1\leq L_2$;
    \item \emph{Cash-invariant}: $R(L-m)=R(L)-m$ for each $m$ in $\mathbb{R}$;
    \item \emph{Sub-additive}: $R(L_1+L_2)\leq R(L_1)+R(L_2)$;
    \item \emph{Positive homogeneous}: $R(\lambda L)=\lambda R(L)$ for every $\lambda\geq 0$.
\end{itemize}
A coherent risk measure is further called \emph{Fatou continuous}, if $R(L)\leq \liminf_n R(L_n)$ whenever $(L_n)$ is a sequence dominated in $L^1$ and converges to $L$ $P$ almost surely.

For $L$ in $L^1$, we consider the quantile based functions 
\begin{enumerate}[label=$\bullet$, fullwidth]
    \item \emph{Value at Risk}: for $0<\alpha <1$,
        \begin{equation*}
            q_{\alpha}(L)=\inf\left\{ m\colon F_L(m)\geq 1-\alpha \right\}.
        \end{equation*}
    \item \emph{Tail Conditional Expectation}: for $0<\alpha <1$,
        \begin{equation*}
            TCE_\alpha(L)=E[L|L>q_\alpha(L)].
        \end{equation*}
    \item \emph{Expected Shortfall}: for $0< \alpha< 1$,
        \begin{equation*}
            ES_{\alpha}(L)=\frac{1}{\alpha}\int_0^\alpha q_u(L)du.
        \end{equation*}
\end{enumerate}
It is well known that the value at risk is not sub-additive -- not even convex -- and therefore not a coherent risk measure, see \citep{tasche2002,foellmer2016}.
While the expected shortfall is coherent and coincides with the  tail conditional expectation  for loss profiles with continuous distributions, in general $ES_\alpha\geq TCE_\alpha$ and $TCE_\alpha$ may not be sub-additive, see \citep{artzner1999,foellmer2016,acerbi2002}.

For risk level $\alpha$ in $(0,1/2]$, the expectile $e_\alpha$ of $L$ in $L^1$ is defined as the unique solution of 
\begin{equation}\label{eq:foc01}
    (1-\alpha)E[(L-e_\alpha(L))^+]=\alpha E[(L-e_\alpha(L))^-].
\end{equation}
It turns out that the expectile is a law invariant, finite valued, and the only elicitable and coherent risk measure, see \citep{weber2006, ziegel2014, bellini2015b, delbaen2016}.
From \citep{bellini2014}, its dual representation is given by
\begin{equation*}
     e_{\alpha}(L) = \max_{Y\in \mathcal{Q}_{\alpha}} E[LY],
\end{equation*}
where
\begin{equation}\label{eq:Expect01}
    \mathcal{Q}_{\alpha}:=\left\{Y\in  \mathcal{Q}\colon \gamma \leq Y\leq \frac{(1-\alpha)\gamma}{\alpha}\text{ for some }\gamma\in \left[\frac{\alpha}{1-\alpha},1\right]\right\}.
\end{equation}
Since expectile can be seen as a generalization of quantiles, if $q_\alpha$ is replaced by $e_\alpha$ in the definition of $TCE_\alpha$ and $ES_\alpha$, then we get the expectile based functions on $L^1$ defined as
\begin{enumerate}[label=$\bullet$, fullwidth]
    \item \emph{Expectile based Tail Conditional Expectation}: for $0<\alpha \leq 1/2$,
        \begin{equation*}
            tce_\alpha(L)=E[L|L>e_\alpha(L)].
        \end{equation*}
    \item \emph{Expectile based Expected Shortfall}: for $0< \alpha\leq 1/2$,
        \begin{equation*}
            es_{\alpha}(L)=\frac{1}{\alpha}\int_{0}^{\alpha}e_u(L) du.
        \end{equation*}
\end{enumerate}
If $L$ is not identically constant, it holds that  $tce_\alpha(L)=ES_{\beta^\ast}(L)$, where $\beta^\ast=P[L>e_\alpha(L)]$.
It is also known that $tce_\alpha$ is not monotone and sub-additive in general and hence, not a coherent risk measure, see \citep{daouia2019}.
However, $es_\alpha$ is coherent and has the following properties.
\begin{proposition}
The expectile based expected shortfall is law invariant, $(-\infty,\infty]$ valued, coherent and Fatou continuous. 
\end{proposition}
\begin{proof}
    It is known that the map  $u\mapsto e_u(L)$ is continuous on $(0,1/2]$, see \citep{newey1987,bellini2014}.
    It implies that $e_\cdot(L)$ is measurable.
    Since $E[L]\leq e_u(L)$ for each $u\in (0,\alpha]$, it holds that the integration is well defined and the range of $es_\alpha$ is a subset of $(-\infty,\infty]$.
    The law invariance and coherent properties of $es_\alpha$ directly follows from expectile.
    Let $(L_n)$ be a sequence dominated in $L^1$ and converging to $L$ almost surely.
    Since a finite valued coherent risk measure is Fatou continuous, it holds that $e_\alpha$ is also Fatou continuous, see \citep{kaina2008}.
    The Fatou continuity of $e_u$ together with Fatou's Lemma yields 
    \begin{multline*}
        es_\alpha(L) = \frac{1}{\alpha}\int_0^\alpha e_u(L)du \leq \frac{1}{\alpha}\int_0^\alpha \left(\liminf_n e_u(L_n)\right)du\\
         \leq \liminf_n \frac{1}{\alpha}\int_0^\alpha e_u(L_n)du=\liminf_n es_\alpha (L_n).
    \end{multline*}
    This ends the proof of the proposition.
\end{proof}

\section{Dual Representation of expectile based expected shortfall}
The expectile based expected shortfall is law-invariant, coherent and Fatou continuous.
Hence, it admits a representation of the form 
\begin{equation*}
    es_\alpha(L)=\sup_{Y\in \mathcal{Q}_{es}}E[LY],
\end{equation*}
for some $\mathcal{Q}_{es}\subseteq \mathcal{Q}$ which is called the dual set of $es_\alpha$, see \citep{kaina2008,biagini2009,cheridito2008}.
This section is dedicated to describe the set $\mathcal{Q}_{es}$.
Throughout this section, we consider the measurable space $(I, \mathcal{I}, \mu)$, where $I=(0,\alpha]$, $\mathcal{I}$ is the Borel sigma algebra of $I$ and $\mu$ is the Lebesgue measure on $\mathcal{I}$.
We denote by $L^0_s(L^\infty)$, the space of all step functions on $I$ with values in $L^\infty$ identified $\mu$  almost every where, that is, 
\begin{equation*}
    L^0_s(L^\infty)=\left\{\sum_{n=1}^\infty L_n 1_{I_n}\colon (L_n) \text{ is a sequence in } L^\infty \text{ and }(I_n)\subseteq \mathcal{I} \text{ is a partition of } I \right\}, 
\end{equation*}
where $1_A$ is the indicator function whose value is $1$ for $u$ in $A$ and $0$, otherwise.
We say a function $\mathbf{Y}\colon I\to L^\infty$ is \emph{measurable} (strongly), if there exist a sequence $(\mathbf{Y}_n)$ in $L^0_s(L^\infty)$ such that $||\mathbf{Y}_n(u)-\mathbf{Y}(u)||_\infty\to 0 $ $\mu$ almost every where .
We also denote by $L^0(L^\infty)$, the spaces of all  measurable functions on $I$ with values in $L^\infty$.
We extend the norm $||\cdot||_\infty$ to  $L^0(L^\infty)$  as
\begin{equation*}
    ||\mathbf{Y}||_\infty(u)=\lim_{n\nearrow \infty}||\mathbf{Y}_n(u)||_\infty,
\end{equation*}
where $(\mathbf{Y}_n)$ is a sequence in $L^0_s(L^\infty)$ such that $||\mathbf{Y}_n(u)-\mathbf{Y}(u)||_\infty \to 0$  $\mu$ almost every where .
Finally, $L^0(I)$ denotes the space of all real valued random variables on $I$ identified $\mu$  almost every where.
Throughout this section, all equalities and inequalities in $L^0(I)$  are identified in the $\mu$  almost every where sense.
Clearly, $u\mapsto ||\mathbf{Y}||_\infty(u)$ and $u\mapsto \langle L,\mathbf{Y}(u)\rangle:=E[L\mathbf{Y}(u)]$  with $L\in L^1$ and $\mathbf{Y}\in L^0(L^\infty)$ are in $L^0(I)$.
It also holds that $L\mapsto e_{\cdot}(L)$ is a function from $L^1$ to $L^0(I)$.

\begin{proposition}\label{prop:almostdual}
    The expectile based expected shortfall admits the representation
    \begin{equation}\label{eq:dual01}
        es_\alpha(L)=\sup_{\mathbf{Y}\in \mathcal{Y}} \frac{1}{\alpha}\int_I E[L\mathbf{Y}]d\mu,
    \end{equation}
    where 
    \begin{multline*}
        \mathcal{Y}=\left\{\mathbf{Y}\in L^0(L^\infty)\colon E[\mathbf{Y}(u)]=1 \text{ and } \gamma(u) \leq \mathbf{Y}(u)\leq \frac{(1-u)\gamma(u)}{u} \right. \\ \left. \text{ for some } \gamma\in L^0(I) \text{ such that } \frac{u}{1-u}\leq \gamma(u)\leq 1 \right\}.
    \end{multline*}

    Furthermore, $\mathbf{Y}^\ast$ in $\mathcal{Y}$  is optimal, if and only if $\mathbf{Y}^\ast(u)$ is the optimal density of $e_u$  for $\mu$-almost all $u$ in $I$.
\end{proposition}

\begin{proof}
    As a result of Relation \ref{eq:Expect01}, it follows that $\mathbf{Y}(u)$ in $\mathcal{Q}_u$ for $\mu$-almost all $u$ in $I$.
    This implies that $e_{\cdot}(L)\geq E[L\mathbf{Y}(\cdot)]$ for all $\mathbf{Y}$ in $\mathcal{Y}$ and therefore 
    \begin{equation}\label{eq:almostdual01}
        es_\alpha(L)\geq \sup_{\mathbf{Y} \in \mathcal{Y}}\frac{1}{\alpha}\int_I E[L\mathbf{Y}] d\mu.
    \end{equation}
    For each $n$ in $\mathbb{N}$, we  consider the partition $\Pi^n$ of $[0,\alpha]$ given by $\Pi^n:=\{t^n_k=k\alpha/n\colon k=0,\dots,n\}$.
    Let
    \begin{equation*}
        \mathcal{Y}^n:=\left\{\mathbf{Y}(u)=\sum_{k=0}^{n-1}Y_{t^n_{k+1}}1_{(t^n_k,t^n_{k+1}]}(u)\colon Y_{t^n_{k+1}}\in \mathcal{Q}_{t^n_{k+1}} \text{for all }k=0,\dots, n-1\right\}.
    \end{equation*}
    For $\mathbf{Y}$ in $\mathcal{Y}^n$, it holds that $\mathbf{Y}\in L^0_s(L^\infty)$ and $E[\mathbf{Y}(u)]=1$.
    Since $Y_{t^n_{k+1}}\in \mathcal{Q}_{t^n_{k+1}}$, there exist $\gamma_{t^n_{k+1}}$ in $\mathbb{R}$ such that 
    \begin{equation*}
        \gamma_{t^n_{k+1}}\in  \left[\frac{t^n_{k+1}}{1-t^n_{k+1}},1\right] \quad \text{and} \quad \gamma_{t^n_{k+1}}\leq Y_{t^n_{k+1}}\leq \frac{(1-t^n_{k+1})\gamma_{t^n_{k+1}}}{t^n_{k+1}}
    \end{equation*}
    for all $k=0,\dots,n-1$.
    Define 
    \begin{equation*}
        \gamma(u):=\sum_{k=0}^{n-1} \gamma_{t^n_{k+1}}1_{(t^n_k, t^n_{k+1}]}(u).
    \end{equation*}
    It follows that $\gamma$ in $L^0(I)$, $(1-u)/u \leq \gamma(u)\leq 1$  and $\gamma(u)\leq \mathbf{Y}(u) \leq (1-u)\gamma(u) /u$.
    Hence, $\mathcal{Y}^n\subseteq \mathcal{Y}$ and Relation \eqref{eq:almostdual01} further implies that 
    \begin{multline}\label{eq:almostdual02}
        es_\alpha(L)\geq \sup_n \left\{\sup_{\mathbf{Y}\in \mathcal{Y}^n} \frac{1}{\alpha}\sum_{k=0}^{n-1} E[LY_{t^n_{k+1}}] (t^n_{k+1}-t^n_k)\right\}=\sup_n \frac{1}{\alpha} \int_I e^n_{\cdot}(L) d\mu,
    \end{multline}
    where $e^n_u(L):=\sum_{k=0}^{n-1} e_{t^n_{k+1}}(L) 1_{(t^n_k, t^n_{k+1}]}(u)$ which is a sequence in $L^0(I)$ such that $e^n_{\cdot}(L)\nearrow e_{\cdot}(L)$ $\mu$ almost everywhere.
    The monotone convergence theorem together with Relation \eqref{eq:almostdual02} yields 
    Relation \eqref{eq:dual01}.

    Let $\mathbf{Y}^\ast$ in $\mathcal{Y}$ be given.
    By the definition of $\mathcal{Y}$, we always have $e_u(L)-E[L\mathbf{Y}^\ast(u)] \geq 0$.
    If $\mathbf{Y}^\ast$ is optimal, then  $\int_I \left(e_{\cdot}(L)-E[L\mathbf{Y}^\ast] \right) d\mu =0$ and hence, $e_{\cdot}(L)=E[L\mathbf{Y}]$.
    The converse statement is clear ending the proof.
\end{proof}

Finally, to provide the dual representations of $es_\alpha$, we need the Bochner integral. 
The step function $\mathbf{Y}=\sum_{n=1}^\infty Y_n 1_{I_n}$ in $L^0_s(L^\infty)$ is  said to be Bochner integrable  with respect to the measure $\mu$, provide that $\int_I ||\mathbf{Y}||_\infty d\mu=\sum_{n=1}^\infty ||Y_n||_\infty\mu(I_n)<\infty$.
In this case, the Bochner integral of $\mathbf{Y}$ is denoted by $\int_I \mathbf{Y}d\mu $ and given by 
\begin{equation*}
    \int_I \mathbf{Y}d\mu\colon =\sum_{n=1}^\infty Y_n \mu(I_n).
\end{equation*}
A function $\mathbf{Y}$ in $L^0(L^\infty)$ is also said to be Bochner integrable, if there exist a Bochner integrable sequence $(\mathbf{Y}_n)$ in $L^0_s(L^\infty)$ such that  $||\mathbf{Y}_n(u)-\mathbf{Y}(u)||_\infty\to 0 $ $\mu$ almost every where  and $\int_I ||\mathbf{Y}_n(u) - \mathbf{Y}(u)||_\infty \mu(du) \to 0$.
In this case, the Bochner integral of $\mathbf{Y}$ with respect to $\mu$ is given by 
\begin{equation*}
    \int_I \mathbf{Y}d\mu:=\lim_{n\nearrow \infty} \int_I\mathbf{Y}_n d\mu.
\end{equation*}
It is well known that $\mathbf{Y}$ in $L^0(L^\infty)$ is Bochner integrable if and only if $\int_I ||\mathbf{Y}(u)||_\infty \mu(du)$ is finite.
For Bochner integrable function $\mathbf{Y}$ in $L^0(L^\infty)$, it also holds that 
\begin{equation}\label{eq:bochner}
    E\left[L \int_I \mathbf{Y}d\mu\right] =\int_I E[L\mathbf{Y}] d\mu, \quad \text{for all }L\in L^1,
\end{equation}
see \citep{hille1957} for instance. 
With this at hands, the dual representation reads as follows.
\begin{theorem}\label{thm:dual}
    The expectile based expected shortfall $es_\alpha$ admits the dual representation
    \begin{equation*}
        es_\alpha(L) = \sup_{Y\in \bar{\mathcal{Q}}_{es}} E[LY],
    \end{equation*}
    where $\bar{\mathcal{Q}}_{es}$ is the $\sigma(L^\infty, L^1)$-closure of the non-empty and convex set 
    \begin{equation*}
        \mathcal{Q}_{es}=\left\{Y\in \mathcal{Q}\colon Y=\frac{1}{\alpha}\int_I\mathbf{Y}d\mu \quad\text{for some } \mathbf{Y}\in \mathcal{Y} \text{ with } \int_I ||\mathbf{Y}||_\infty d\mu <\infty \right\}.
    \end{equation*}
    Furthermore, $Y^\ast\in \mathcal{Q}_{es}$ is optimal for $es_\alpha$ if and only if  $\mathbf{Y}^\ast$ in $\mathcal{Y}$ for which $Y^\ast =\frac{1}{\alpha}\int_I \mathbf{Y}^\ast d\mu$ is optimal for Relation \eqref{eq:dual01}.
\end{theorem}
\begin{proof}
    Let $Y$ be in $\mathcal{Q}_{es}$, that is, $ Y=\frac{1}{\alpha} \int_I \mathbf{Y} d\mu$ for some Bochner integrable function $\mathbf{Y}$ in $\mathcal{Y}$.
    Relation \eqref{eq:bochner} and Proposition \ref{prop:almostdual} yields 
    \begin{equation*}
        E[LY]=E\left[L\frac{1}{\alpha}\int_I \mathbf{Y} d\mu\right]=\frac{1}{\alpha} \int_I E[L\mathbf{Y}]d\mu \leq es_\alpha(L).
    \end{equation*}
    It follows that
    \begin{equation}\label{eq:dual02}
        \sup_{Y\in \mathcal{Q}_{es}}E[LY] \leq es_\alpha(L).
    \end{equation}
    For each $\mathbf{Y}$ in $\mathcal{Y}^n$, it holds that
    \begin{equation*}
        \int_I ||\mathbf{Y}||_\infty d\mu =\sum_{k=0}^{n-1} ||Y_{t^n_{k+1}}||_\infty (t^n_{k+1}-t^n_k) =  \frac{\alpha}{n}\sum_{k=0}^{n-1} ||Y_{t^n_{k+1}}||_\infty <\infty,
    \end{equation*} 
    implying that every element of $\mathcal{Y}^n$ is Bochner integrable.
    Hence, Relation \eqref{eq:almostdual02} and \eqref{eq:bochner} yields 
    \begin{equation}\label{eq:dual03}
        es_\alpha(L)= \sup_n \left\{ \sup_{Y\in \mathcal{Q}^n_{es}} E[LY]\right\}
        \leq \sup_{Y\in \mathcal{Q}_{es}} E[LY], 
    \end{equation}
    where 
    \begin{equation*}
        \mathcal{Q}^n_{es}=\left\{Y\in \mathcal{Q}\colon Y=\frac{1}{\alpha} \int_I \mathbf{Y} d\mu \text{ for some }\quad \mathbf{Y} \in \mathcal{Y}^n \right\}.
    \end{equation*}
    The last inequality follows from the fact that $\mathcal{Q}^n_{es}\subseteq \mathcal{Q}_{es}$ for all $n$ in $\mathbb{N}$.
    Relation \eqref{eq:dual02} and \eqref{eq:dual03} yields 
    \begin{equation*}
        es_\alpha(L)=\sup_{Y\in \mathcal{Q}_{es}} E[LY].
    \end{equation*}
    Clearly, $\mathcal{Q}_{es}$ is non-empty and convex subset of $\mathcal{Q}$.
    Hence, taking the $\sigma(L^\infty, L^1)$-closure $\mathcal{Q}_{es}$ do not affect the supremum.
   The last assertion directly follows from Proposition \ref{prop:almostdual}.
\end{proof}

\section{Properties of Expectile based expected shortfall }
\subsection{Comonotonicity}
A coherent risk measure $R:L^1\to (-\infty, \infty]$ is said to be \emph{comonotonic}, if
$R(L_1+L_2)= R(L_1)+R(L_2)$ for each comonotone pairs \footnote{We say $L_1$ and $L_2$ in $L^1$ are \emph{comonotone}, if $( L_1(\omega)-L_1(\omega'))( L_2(\omega)-L_2(\omega'))\geq 0$ for all $(\omega,\omega')\in \Omega\times \Omega$.} of loss profiles $L_1$ and $L_2$.
It is well known that the quantile based expected shortfall is comonotonic, while the expectile is not, see \citep{shapiro2013,foellmer2016,emmer2015} for instance.
It is therefore not astonishing that the expectile based expected shortfall is not comonotonic as shown in the following example.
\begin{example}
    Let $\varphi$ be the concave distortion function given by 
    \begin{equation}\label{eq:primitive}
        \varphi(t)=
        \begin{cases}
            0& \text{ if }t=0\\
            1-\frac{\alpha}{2} & \text{ if }t=1/2\\
            -\frac{t}{1-2t}\left[1-\frac{1-t}{\alpha(1-2t)}\ln{\left(1-2\alpha+\frac{\alpha}{t}\right)}\right] & \text{ if }t\neq 1/2\text{ and } 0
        \end{cases}.
    \end{equation}
    Then by Example \ref{eg:bern}, $\varphi$ corresponds to the distortion probability measure  $C_\varphi(A)=es_\alpha(1_A)$.
    For $L\sim Unif[0,1]$, following Example \ref{eg:unif}, we get $es_{0.34}(L)=0.706$.
    The concave distortion function $R_\varphi(L)=\int_{(0,1]} \varphi(P[L>x])dx=\int_{(0,1]} \varphi(1-x)dx=0.758 $.
    Hence, $R_\varphi(L)\neq es_{0.34}(L)$ and  by \citep[Corollary 4.95]{foellmer2016} $es_\alpha$ is not comonotonic.
\end{example}

In the spirit of \citep[Theorem 6]{delbaen2013}, the following proposition provides the smallest comonotonic risk measure that dominates  $es_\alpha$ uniformly on  $L^1$.

\begin{proposition}\label{prop:comonoton}
    Let $\varphi$ be the distortion function given by Relation \eqref{eq:primitive}, it holds that
    \begin{equation*}
        es_\alpha(L)\leq R_{\varphi}(L):=\int_{(0,1]} \varphi_+'(t)q_t(L) dt.
    \end{equation*}
    Moreover, $R_{\varphi}$ is the smallest law invariant coherent and comonotonic risk measure dominating $es_\alpha$ uniformly for each $L$ in $L^1$.
    In particular, $es_\alpha(1_A)=R_{\varphi}(1_A)=\varphi(P[A])$ for each $A\in \mathcal{F}$.
\end{proposition}
\begin{proof}
    Following  \citep[Theorem 6]{delbaen2013}, for each $u$ in $(0,1/2]$, $e_u$ is dominated uniformly for each $L$ in $L^1$ by the smallest law-invariant, coherent and comonotonic risk measure  as:
    \begin{equation*}
        e_u(L)\leq \int_{(0,1]}\frac{(u(1-u)}{((1-2u)t+u)^2}q_t(L)dt.
    \end{equation*}
    Using Fubini's theorem yields 
    \begin{multline*}             
        es_\alpha(L) \leq \frac{1}{\alpha}\int_0^\alpha \left(\int_{(0,1]}\frac{u(1-u)}{((1-2u)t+u)^2}q_t(L) dt\right)du\\
        =\int_{(0,1]}\left(\frac{1}{\alpha}\int_0^\alpha \frac{u(1-u)}{((1-2u)t+u)^2} du\right)q_t(L) dt
        =\int_{(0,1]} \varphi_+'(t) q_t(L) dt=R_\varphi(L),
    \end{multline*}
    where
    \begin{equation*}
        \varphi_+'(t)=\begin{cases}
        \frac{-1}{(1-2t)^2}\left[\frac{1+(1-2t)\alpha}{t+(1-2t)\alpha}-\frac{1}{\alpha(1-2t)}\ln\left(1-2\alpha+\frac{\alpha}{t}\right)\right] & \text{ for }t\neq 1/2\\
        2\alpha\left(1-\frac{2}{3}\alpha\right) & \text{for } t=1/2
        \end{cases}.
    \end{equation*}
    Clearly, $R_\varphi$ is the smallest and law-invariant, coherent and comonotonic risk measure that dominates $es_\alpha$ uniformly.
    From Example \ref{eg:bern}, we also have  $es_\alpha(1_A)=\varphi(P[A])$.
\end{proof}

\subsection{Quantile Versus expectile based expected shortfall}\label{sec:sec42}
For a given risk level $\alpha$ in $(0,1/2]$, the expectile is less conservative than expectile based expected shortfall, that is, $e_\alpha\leq es_\alpha$, see \citep{daouia2019} for instance. However, as compared to the quantile based expected shortfall, it holds that $es_\alpha\leq ES_\alpha$ or $es_\alpha>ES_\alpha$ depending on the considered loss profile, see figure \ref{fig:fig01}.
\begin{figure}[H]
    \centering
    \begin{subfigure}{.49\textwidth}
        \includegraphics[width=\textwidth]{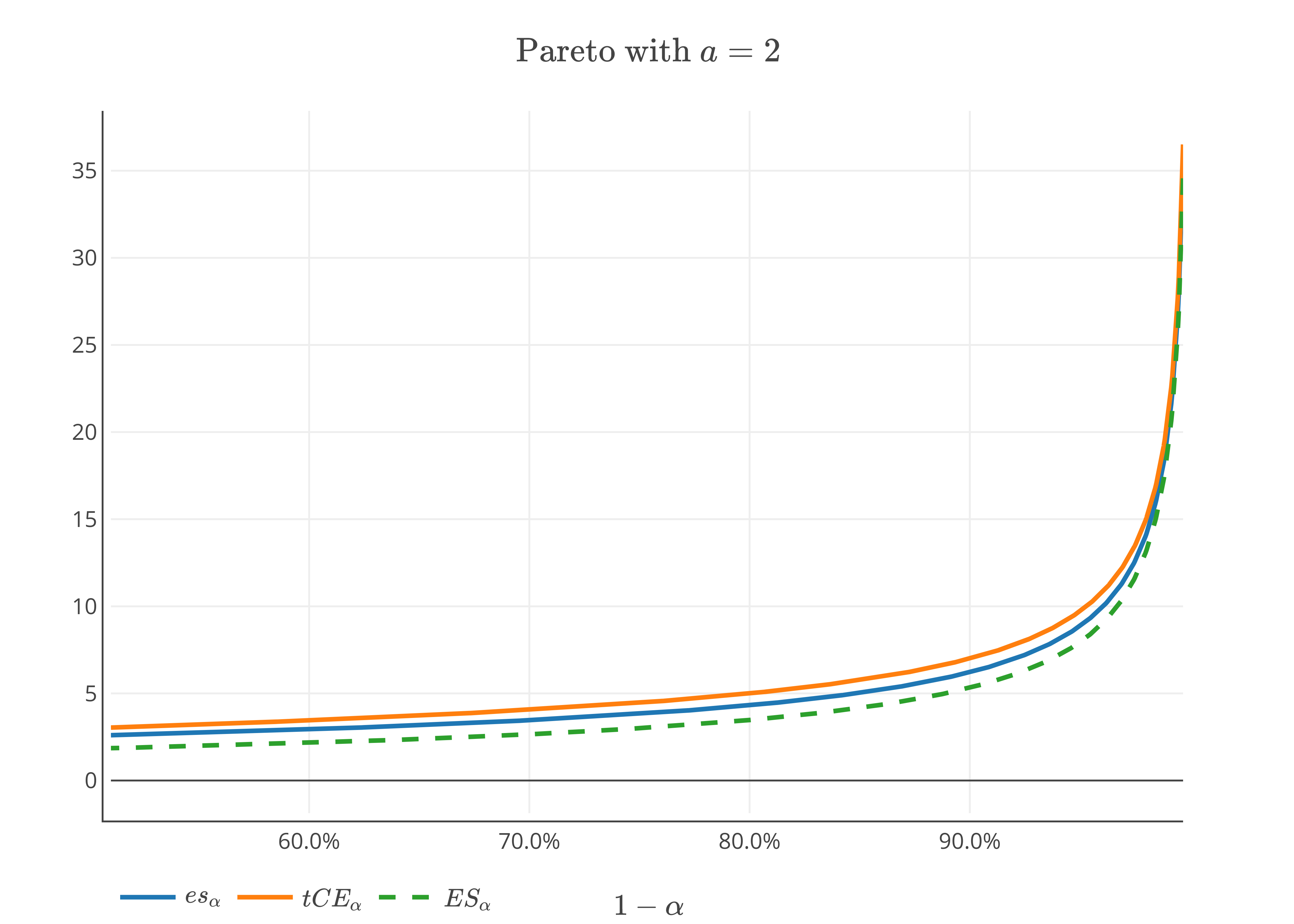}
    \end{subfigure}
    \begin{subfigure}{.49\textwidth}
        \includegraphics[width=\textwidth]{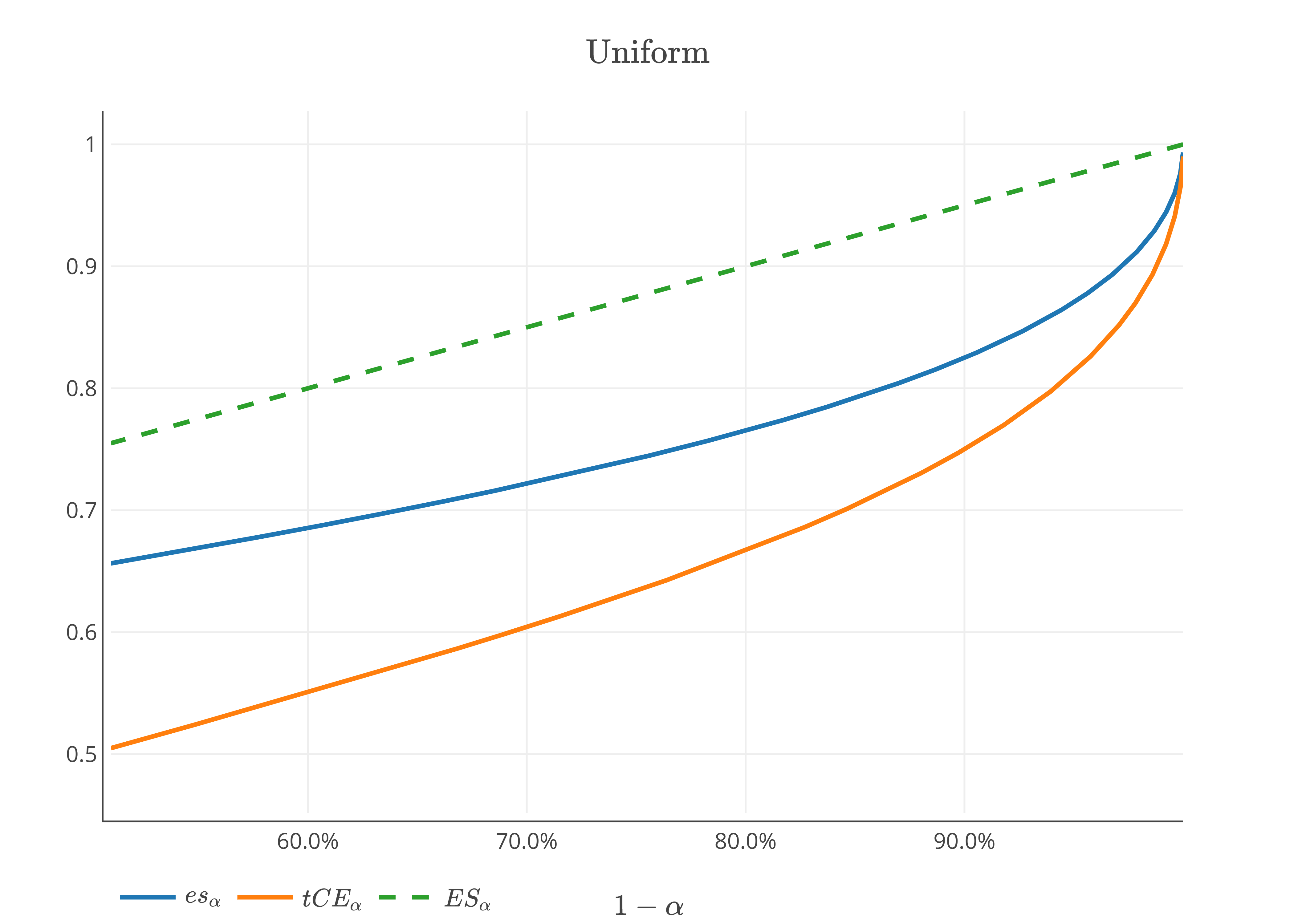}
    \end{subfigure}
    \begin{subfigure}{.49\textwidth}
        \includegraphics[width=\textwidth]{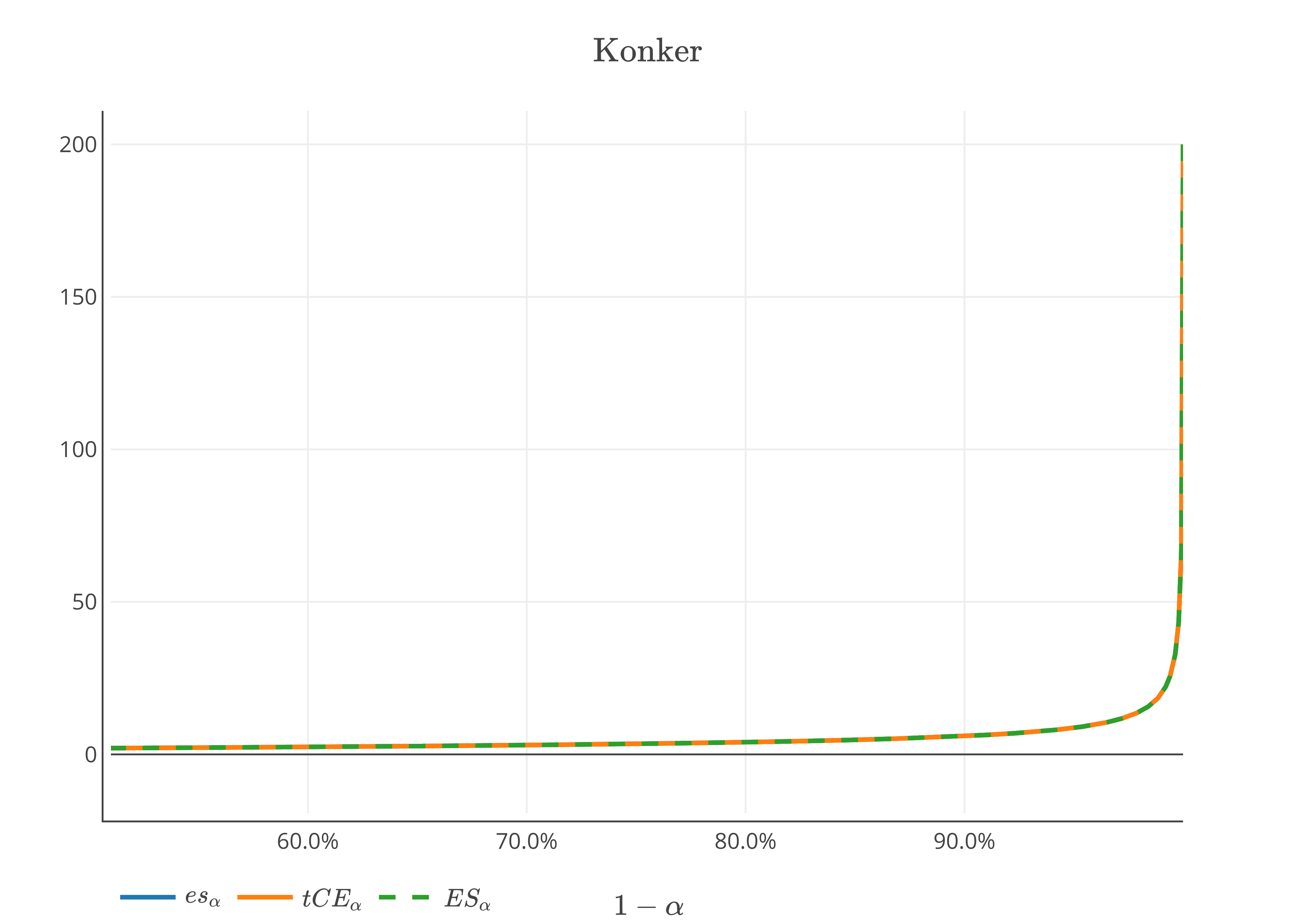}
    \end{subfigure}
    \begin{subfigure}{.49\textwidth}
        \includegraphics[width=\textwidth]{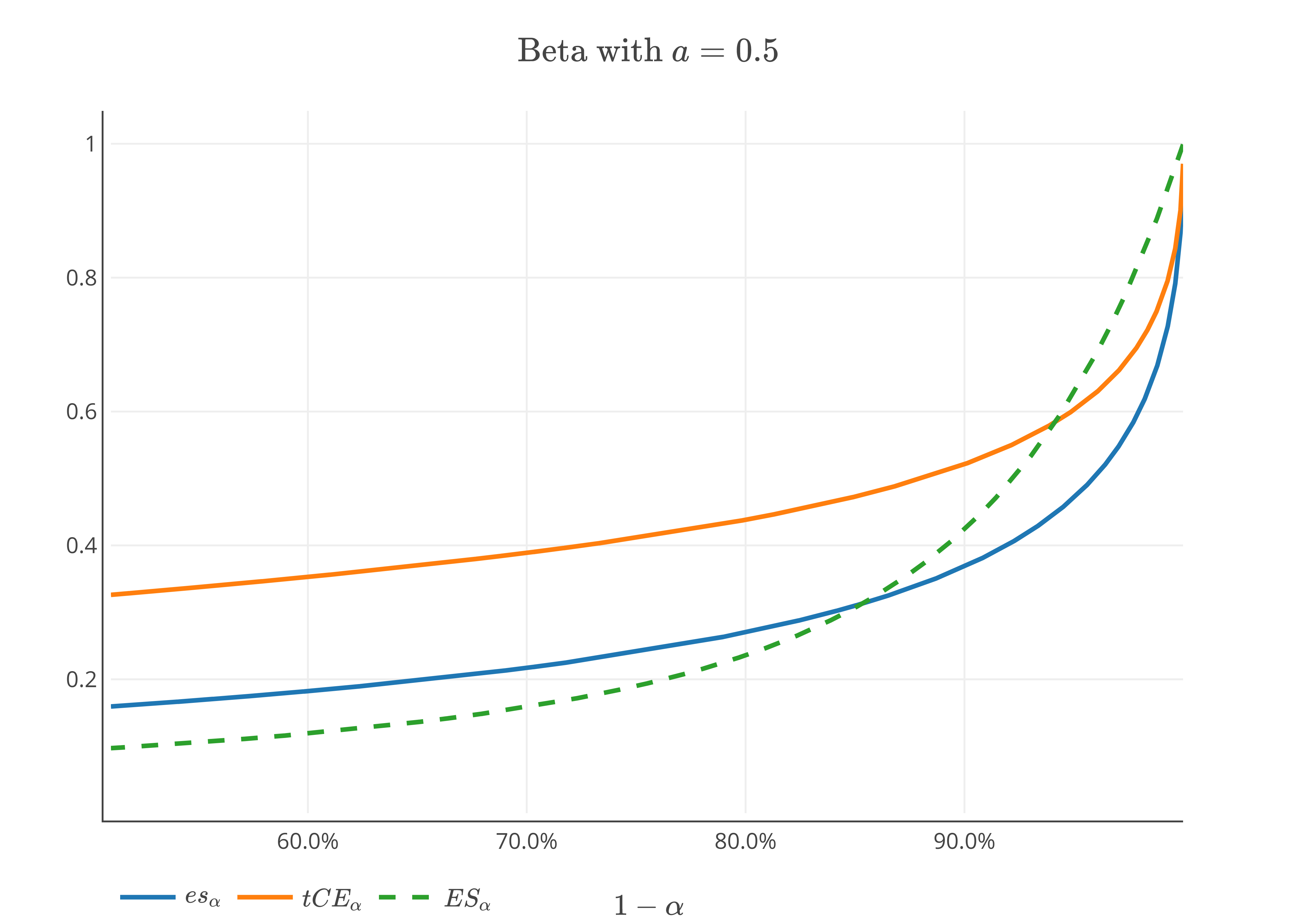}
    \end{subfigure}
    \caption{Graphs of $es_\alpha$, $tce_\alpha$ and $ES_\alpha$ for Pareto, uniform, Konker and beta distributions.}
    \label{fig:fig01}
\end{figure}
Using the lower bounds of expectile given in \citep[Proposition 3.1]{mekonnen2019}, we provide a family of lower bounds for $es_\alpha$ in terms of convex combination of expected shortfalls as follows. 
\begin{proposition}\label{prop:lowerbound}
    For each $\beta$ in $(0,1)$, it holds that 
    \begin{equation*}
        (1-\gamma_\beta)ES_\beta(L)+ \gamma_\beta E[L] \leq es_\alpha(L),
    \end{equation*}
    where 
    \begin{equation}\label{eq:gamma}
        \gamma_\beta:=\begin{cases}
            \alpha & \text{ for } \beta=1/2\\
            \frac{1}{(1-2\beta)}-\frac{\beta}{\alpha(1-2\beta)^2 } \ln{\left(1-2\alpha+\frac{\alpha}{\beta}\right)} & \text{ for } \beta\neq 1/2
        \end{cases}.
    \end{equation}
    Furthermore, the risk measure $R_\alpha$ defined as
    \begin{equation*}
        R_\alpha(L):=\sup_{0<\beta<1} (1-\gamma_\beta)ES_\beta(L)+ \gamma_\beta E[L], \quad L\in L^1
    \end{equation*}
    is law invariant and coherent such that $R_\alpha(L)\leq es_\alpha(L)$ uniformly for $L$ in $L^1$ and $R_\alpha(1_A)=es_\alpha(1_A)$ for every $A\in \mathcal{F}$.
\end{proposition}
\begin{proof}
    Let $u$ be in $(0,\alpha]$ be given.
    From \citep[Proposition 3.1]{mekonnen2019}, for each $\beta$ in $(0,1)$ we get 
   \begin{equation*}
       \left(1-\frac{u}{(1-2u)\beta+\alpha}\right)ES_\beta(L)+\frac{u}{(1-2u)\beta+\alpha}E[L]\leq e_u(L).
   \end{equation*}
   Integrating both sides of the above inequality with respect to $u$ gives the first result of the proposition.
   The law invariant and coherent property of $R_\alpha$ directly follows from the properties of  $ES_\beta$.
   Let $A$ be in $\mathcal{F}$ such that $P[A]=p$ with $0<p<1$.
   A simple computation using $\beta=p$  yields $(1-\gamma_p)ES_p+\gamma_p E[1_A]=es_\alpha(1_A)$.
   Hence, $R_\alpha(1_A)=es_\alpha(1_A)$.
\end{proof}
\begin{remark}
    Note that the inequality $R_\alpha \leq es_{\alpha}$ can be strict as shown in Example \ref{eg:konker}.
\end{remark}

The expectile can be uniformly dominated by the convex combination of expected shortfalls, see \citep{mekonnen2019} for instance. 
However, this is not the case for the expectile based expected shortfall as shown by the following proposition. 
\begin{proposition}
    The expectile based expected shortfall $es_\alpha$ can not be dominated uniformly for $L$ in $L^1$ by coherent risk measures of the form
    \begin{equation*}
        (1-\lambda)ES_\beta(L)+\lambda ES_\delta(L),
    \end{equation*}
    for some $0<\beta\leq 1$, $0\leq \lambda\leq 1$ and $0<\delta\leq 1$.
\end{proposition}

\begin{proof}
    We know that the concave distortion  function that corresponds to
    \begin{equation*}
        (1-\lambda)ES_\beta(L)+\lambda ES_\delta(L)
    \end{equation*}
    is given by
    $\varphi_{\lambda,\beta,\delta}(t)=(1-\lambda)(t/\beta\wedge 1)+\lambda (t/\delta \wedge 1)$.
    Furthermore, the tangent to $\varphi$ at $(0,0)$ and $(1,1)$ is given by $x=0$ and $y=\gamma_1 x+(1-\gamma_1)$, respectively.
    Suppose there exist a risk measure of the form $(1-\hat{\lambda})ES_{\hat{\beta}}+\hat{\lambda} ES_{\hat{\delta}}$ dominating $es_\alpha$ for some $0<\hat{\beta}, \hat{\delta}\leq 1$ and $0\leq \hat{\lambda}\leq 1$.
    It follows that $\varphi(t)\leq \varphi_{\hat{\lambda},\hat{\beta},\hat{\delta}}(t)$ for each $t$ in $[0,1]$.
    With out loss of generality we take the smallest one from these bounds, that is, $\varphi_{\hat{\lambda},\hat{\beta},\hat{\delta}}$ is tangent to the graph of $\varphi$ at $(0,0)$ and $(1,1)$.
    This contradict the fact that the tangent to $\varphi$ at $(0,0)$ is the $x$-axis.
    Hence, our supposition is false and hence, there is no such upper bound, see Figure \ref{fig:primitive}.
    \begin{figure}[H]
        \centering
        \includegraphics[width=0.8\textwidth]{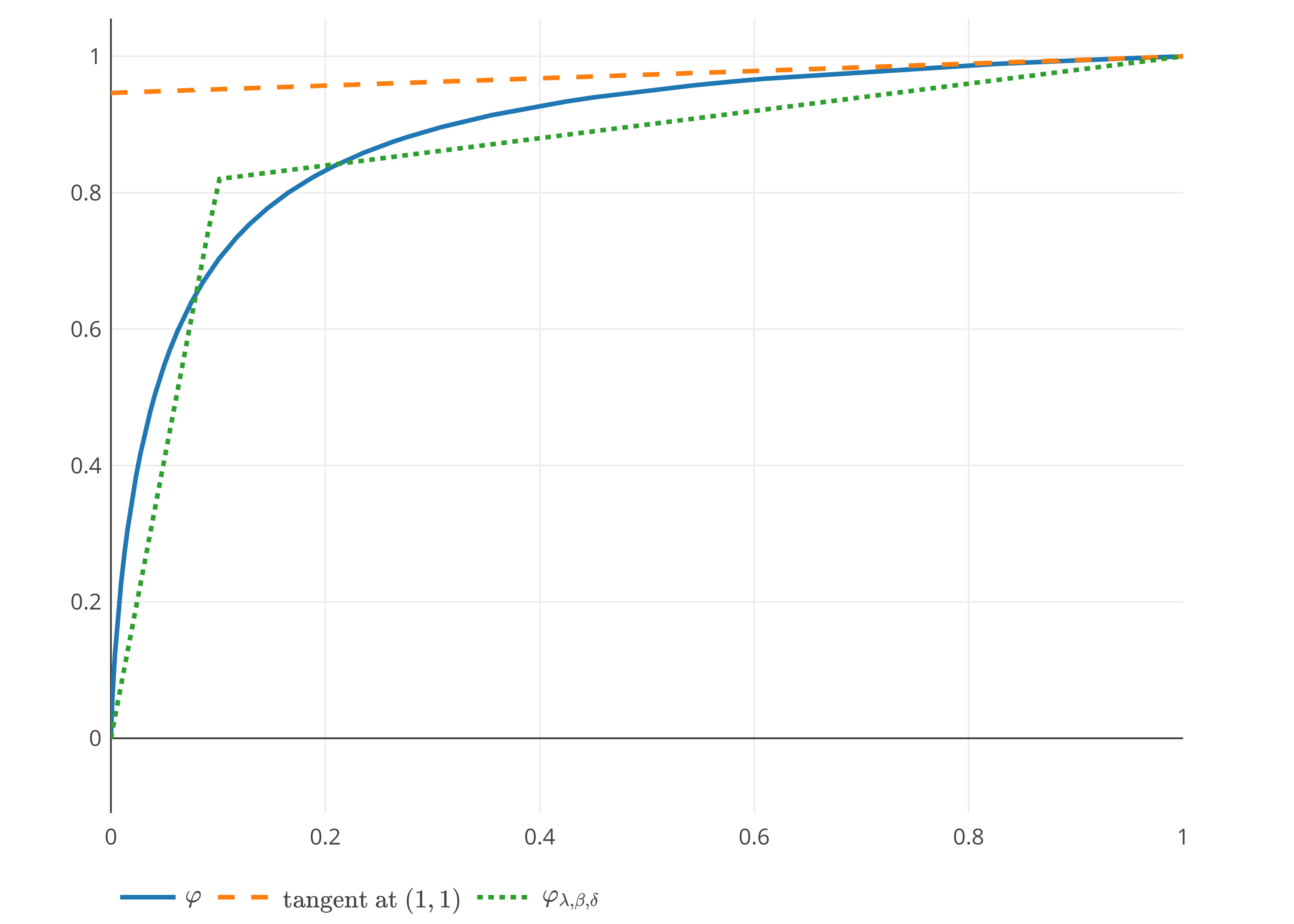}
        \caption{Graph of $\varphi$ and $\varphi_{\lambda,\beta,\delta}$ for $\alpha=10\%, \lambda=20\%$, $\beta=10\%$ and $\delta=100\%$.}
        \label{fig:primitive}
    \end{figure}
\end{proof}

\subsection{Asymptotic Behavior of expectile based expected shortfall}
For a given risk level $\alpha$, the quantile based   tail conditional expectation and expected shortfall coincides, that is $TCE_\alpha=ES_\alpha$, provided that the distribution of $L$ is continuous, see \citep{tasche2002, foellmer2016}. 
However, this is not true in general for $tce_{\alpha}$ and $es_{\alpha}$ neither of both even dominating the other depending on the considered loss profile, see Figure \ref{fig:fig01}.
As discussed in Section \ref{sec:sec42}, it also holds that $ES_\alpha>es_\alpha$ or $ES_\alpha\leq es_\alpha$, depending on the considered loss profile.
When $F_L$ is attracted by extreme value distribution, \citep{bellini2015a} and \citep{mao2015} provide an asymptotic relationship between value at risk and expectile, \citep{tang2012} gives an asymptotic behavior of expected shortfall in terms of value at risk,  \citep{mekonnen2019} provides the asymptotic behavior of expectile in terms expected shortfall, and \citep{daouia2019} also study the asymptotic behavior of $es_\alpha$ in terms of $tce_\alpha$ and $ES_\alpha$ for Fr\'{e}chet type distributions.
Following these results, we are interested to provide the asymptotic behavior of $es_\alpha$ with respect to $tce_\alpha$ and $ES_\alpha$  for loss profiles attracted \footnote{We say $F_L$ is attracted by an extreme value distribution function $H$ and  denoted by $MDA(H)$, if there exist constants $c_n>0$ and $d_n\in \mathbb{R}$ for each $n$ in $\mathbb{N}$ such that
\begin{equation*}
    \lim_{n\nearrow \infty} F_L^n(c_n x+d_n)=H(x).
\end{equation*}}
 by extreme value distribution $H$.
It is well known that $H$ can only be either Weibull ($\Psi_\eta$), Gumbel ($\Lambda$) or Fr\'{e}chet  ($\Phi_\eta$) with parameter $\eta>0$, where $\Psi_\eta(x)=\exp(-(-x)^\eta)$ for $x<0$,
$\Lambda(x)=\exp(-e^{-x})$ for $x \in \mathbb{R}$, and $\Phi_\eta(x)=\exp(-x^{-\eta})$ for $x>0$, see \citep{neil2015,dehaan2006} for instance.

The following proposition states the asymptotic relationship between $es_\alpha$, $tce_\alpha$ and $ES_\alpha$ based on each classes of extreme value distributions.
\begin{proposition}\label{prop:MDA}
    Let $\hat{x}:=\sup\{m\in \mathbb{R}\colon F_L(m)<1\}$. If $0<\hat{x}\leq \infty$, as the risk level $\alpha$ goes to $0$, it holds that 
    \begin{enumerate}[label=($\roman*$), fullwidth]
        \item (\emph{Fr\'{e}chet:}) If $F_L$ is in $MDA(\Phi_\eta)$ with $\eta>1$,  \begin{equation*}
            es_\alpha(L) \sim tce_\alpha(L) \sim (\eta-1)^{-\frac{1}{\eta}}ES_\alpha(L) \quad \text{and} \quad \frac{tce_\alpha(L)}{e_\alpha(L)}  \sim  \frac{ES_\alpha(L)}{q_\alpha(L)}.
        \end{equation*}
        \item (\emph{Weibull:}) If $F_L$ is  in $MDA(\Psi_\eta)$ with $\eta>0$,
        \begin{equation*}
            \frac{\hat{x}-ES_\alpha(L)}{\hat{x}-es_\alpha(L)}=o(1)\quad \text{and} \quad \frac{\hat{x}-tce_\alpha(L)}{\hat{x}-e_\alpha(L)}  \sim  \frac{\hat{x}-ES_\alpha(L)}{\hat{x}-q_\alpha(L)}.
        \end{equation*}
        \item (\emph{Gumbel:}) If $F_L$ is  in  $MDA(\Lambda)$,
        \begin{equation*}
            \frac{tce_\alpha(L)}{e_\alpha(L)}  \sim  \frac{ES_\alpha(L)}{q_\alpha(L)} \quad \text{and} \quad \frac{\hat{x}-tce_\alpha(L)}{\hat{x}-e_\alpha(L)}  \sim  \frac{\hat{x}-ES_\alpha(L)}{\hat{x}-q_\alpha(L)}
        \end{equation*}
        for the case $\hat{x}=\infty$ and $\hat{x}<\infty$, respectively.

        If further $F_L(x)=1-\exp{(-x^\tau r(x))}$ for some slowly varying function \footnote{A measurable function $r:\mathbb{R} \to \mathbb{R}$ is said to be slowly varying, if $\lim_{t\nearrow \infty}\frac{r(tx)}{r(t)}=1$, for each $x$ in $\mathbb{R}$. } $r$ and constant $\tau>0$ such that
    \begin{equation}\label{eq:mda04}
        \lim_{x \nearrow \infty}\left(\frac{r(c x)}{r(x)}-1\right) \ln{r(x)}=0
    \end{equation}
    for some constant $c>0$, it holds that
    \begin{equation*}
        es_\alpha(L) \sim tce_\alpha(L) \sim  ES_\alpha(L).
    \end{equation*}
    \end{enumerate}
\end{proposition}
\begin{proof}
    The Fr\'{e}chet  case directly follows from \citep[Proposition 3]{daouia2019} and \citep[Proposition 4.1]{mekonnen2019}.
    For Weibull case, from \citep[Proposition 3.3]{mao2015}, we get $\hat{x}-q_\alpha(L)=o(\hat{x}-e_\alpha(L))$ as $\alpha$ goes to $0$.
    It follows that
    \begin{equation*}
        \hat{x}-ES_\alpha(L)=\frac{1}{\alpha}\int_0^\alpha (\hat{x}-q_u(L))du =\frac{1}{\alpha}\int_0^\alpha (\hat{x}-e_u(L))g(u)du,
    \end{equation*}
    for some function $g$ such that $g(u)$ goes to $0$ as $u$ goes to $0$.
    Hence,
    \begin{equation*}
        \hat{x}-ES_\alpha(L)\sim \frac{1}{\alpha}\int_0^\alpha (\hat{x}-e_u(L))g(\alpha)du=g(\alpha)(\hat{x}-es_\alpha(L)).
    \end{equation*}
    That is, $\hat{x}-ES_\alpha(L)=o(\hat{x}-es_\alpha(L))$.
    Since $tce_\alpha(L)=ES_{\beta^\ast}(L)$, it follows that 
    \begin{equation*}
        tce_\alpha(L)=q_{\beta^\ast}(L)+\frac{1}{\beta^\ast} E[(L-q_{\beta^\ast}(L))^+].
    \end{equation*}
    Rearranging gives 
    \begin{equation}\label{eq:mda01}
        \frac{tce_\alpha(L)}{e_\alpha(L)}=1+\frac{1}{\beta^\ast e_\alpha(L)} E[(L-e_\alpha(L))^+].
    \end{equation}
    Relation \eqref{eq:mda01} can be re-written as
    \begin{equation}\label{eq:mda02}
        \frac{\hat{x}-tce_\alpha(L)}{\hat{x}-e_\alpha(L)}=1-\frac{1}{\beta^\ast (\hat{x}-e_\alpha(L))} E[(L-e_\alpha(L))^+].
    \end{equation}
    From  [Lemma 3.2]\cite{mao2012}, as $x$ goes to $\hat{x}$ we have that
    \begin{equation*}
        \frac{1}{(\hat{x}-x)(1-F_L(x))} E[(L-x)^+]\sim \frac{1}{\eta+1}.
    \end{equation*}
    Since as $\alpha$ goes to $0$, we have  $e_\alpha$ goes to $\hat{x}$, it follows that
    \begin{equation*}
        \frac{1}{(\hat{x}-e_\alpha(L))\beta^\ast} E[(L-e_\alpha(L))^+]\sim \frac{1}{\eta+1}.
    \end{equation*}
    Hence, Relation \eqref{eq:mda02} yields
    \begin{equation*}
        \frac{\hat{x}-tce_\alpha(L)}{\hat{x}-e_\alpha(L)}\sim 1-\frac{1}{\eta+1}=\frac{\eta}{\eta+1}\sim \frac{\hat{x}-ES_\alpha(L)}{\hat{x}-q_\alpha(L)}.
    \end{equation*}
    The relation for $(\hat{x}-ES_\alpha(L))/(\hat{x}-q_\alpha(L))$ is due to \citep[Theorem 3.4]{mao2012}.

    As for the Gumbel case,
    from  \citet[Relation 3.20]{mao2015}, we have
    \begin{equation}\label{eq:mda03}
       E[(L-e_\alpha(L))^+]=\begin{cases}
           e_\alpha(L) o(\beta^\ast),& \hat{x}=\infty\\
           (\hat{x}-e_\alpha(L))o(\beta^\ast), & \hat{x}<\infty
       \end{cases}.
    \end{equation}
    The Relations \eqref{eq:mda01}, \eqref{eq:mda02} and \eqref{eq:mda03} together with \citep[Theorem 3.4]{mao2012}  yields
    \begin{equation*}
        \frac{tce_\alpha(L)}{e_\alpha(L)}\sim 1\sim \frac{ES_\alpha(L)}{q_\alpha(L)} \quad \text{and}\quad \frac{\hat{x}-tce_\alpha(L)}{\hat{x}-e_\alpha(L)}\sim 1\sim \frac{\hat{x}-ES_\alpha(L)}{\hat{x}-q_\alpha(L)}
    \end{equation*}
    for the case $\hat{x}=\infty$ and $\hat{x}<\infty$, respectively.
    The last asymptotic result holds as a result of $e_\alpha(L) \sim q_\alpha(L)$ under the given conditions, see \citep[Proposition 2.4]{bellini2015a}.
    This ends the proof the proposition.
\end{proof}

According to Proposition \ref{prop:MDA}, for Fr\'{e}chet and Gumbel (with condition \eqref{eq:mda04} ) cases, as the risk level $\alpha$ goes to $0$, the expectile based  tail conditional expectation and expectile based shortfall are equivalent.
In this case, both $tce_\alpha$ and $es_\alpha$ can be interpreted as the expectation of the loss under the event that the loss exceeds $e_\alpha$.
Figure \ref{fig:fig03} provides a graphical illustration for the ratio of $ES_\alpha/es_\alpha$ for Pareto, uniform, beta and exponential distributions.
Notice that the Pareto distribution with parameter $a$ is attracted by  Fr\'{e}chet type $MDA(\Phi_a)$, the uniform and beta distributions are attracted by Weibull type $MDA(\Psi_1)$ and the exponential distribution is attracted by the Gumbel type $MDA(\Lambda)$.
\begin{figure}[H]
    \centering
    \begin{subfigure}{.49\textwidth}
        \includegraphics[width=\textwidth]{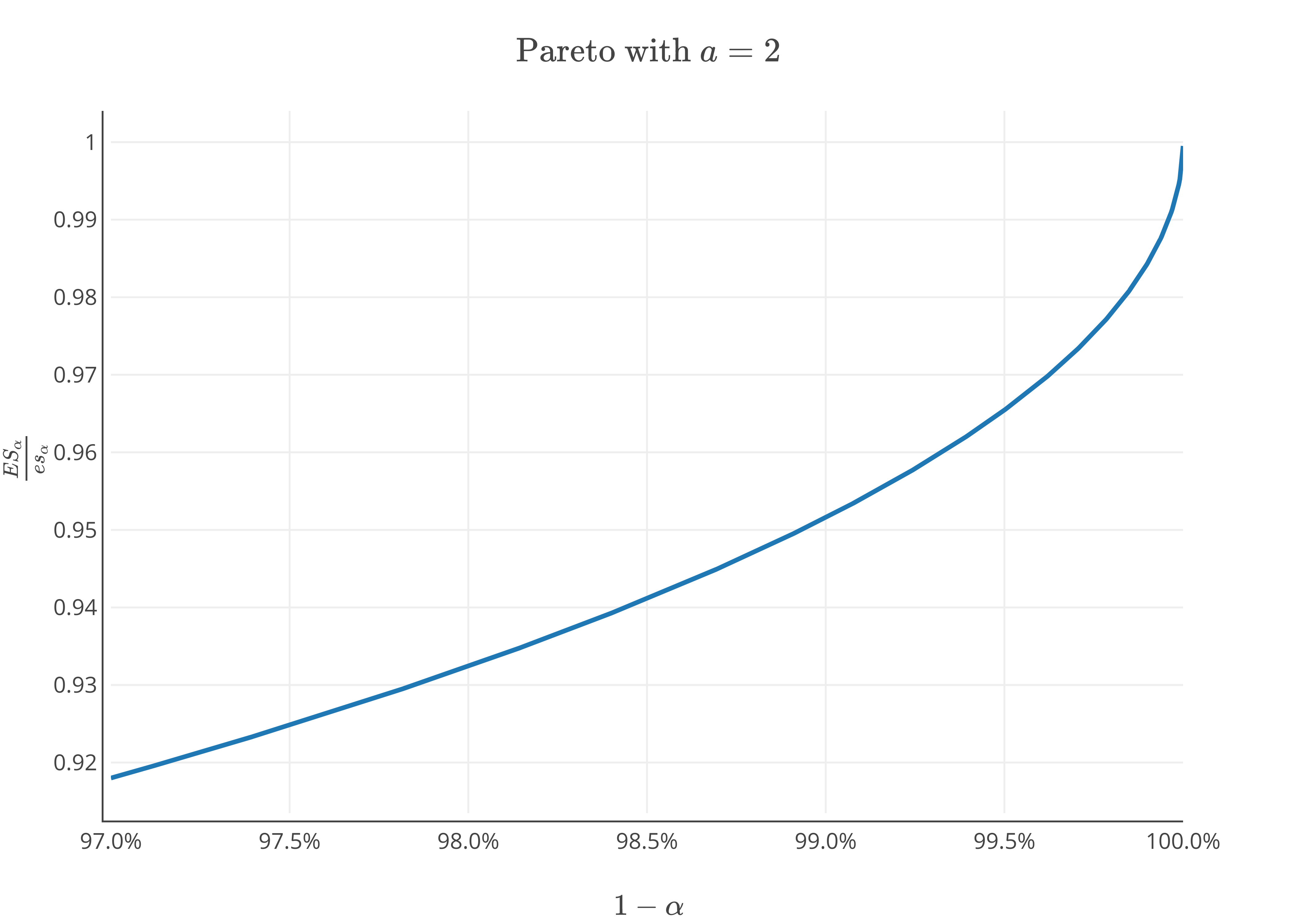}
    \end{subfigure}
    \begin{subfigure}{.49\textwidth}
        \includegraphics[width=\textwidth]{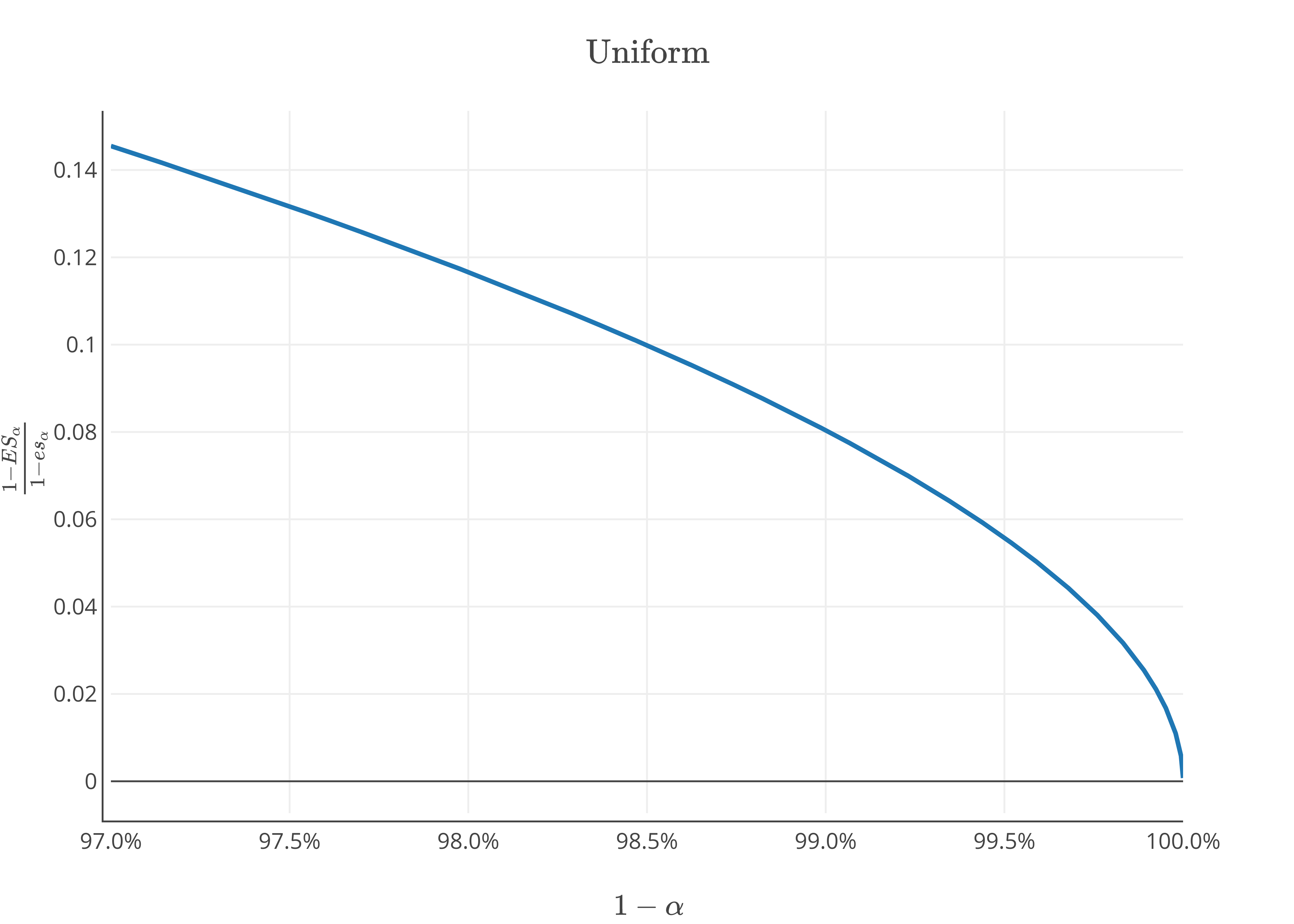}
    \end{subfigure}
    \begin{subfigure}{.49\textwidth}
        \includegraphics[width=\textwidth]{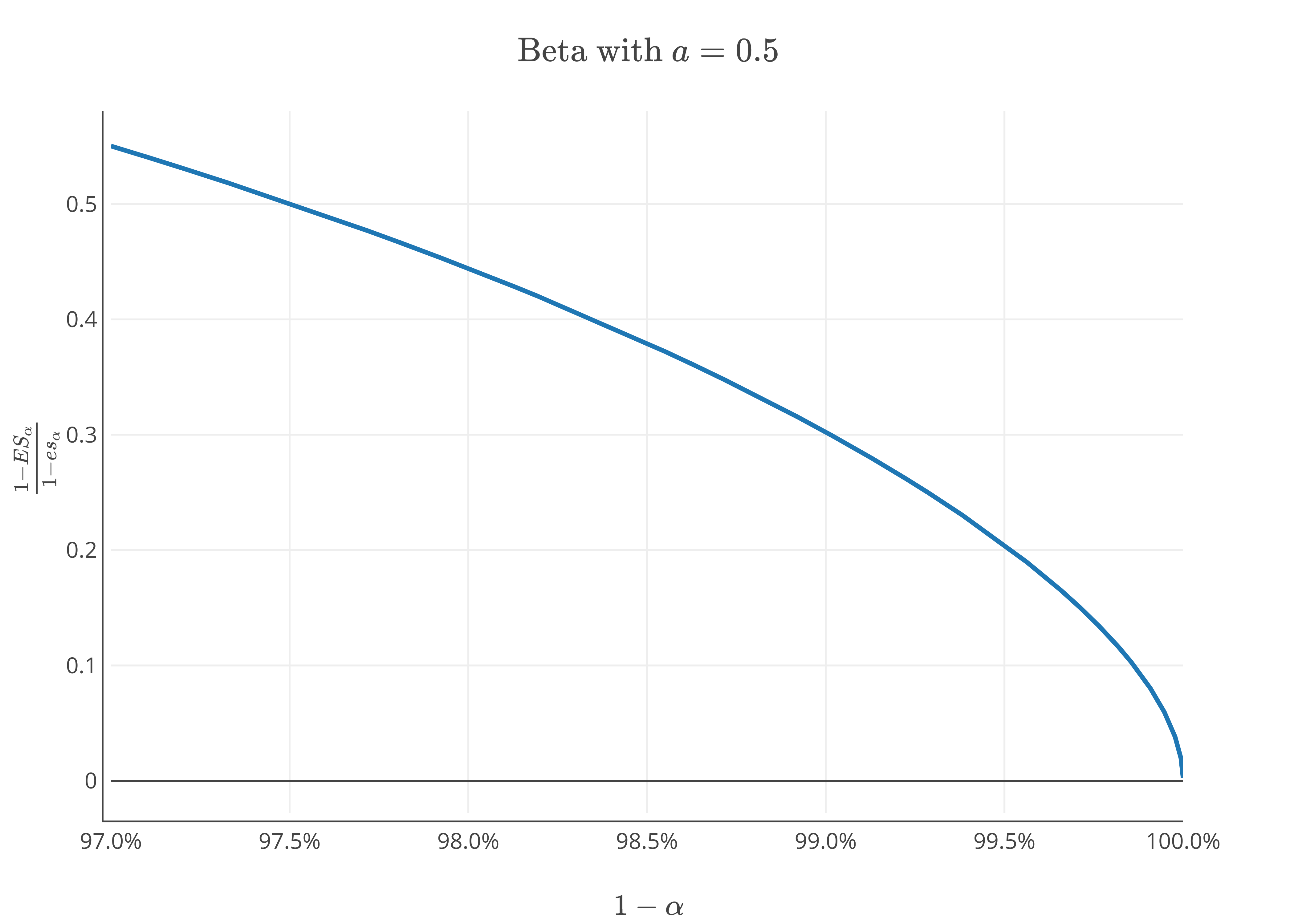}
    \end{subfigure}
    \begin{subfigure}{.49\textwidth}
        \includegraphics[width=\textwidth]{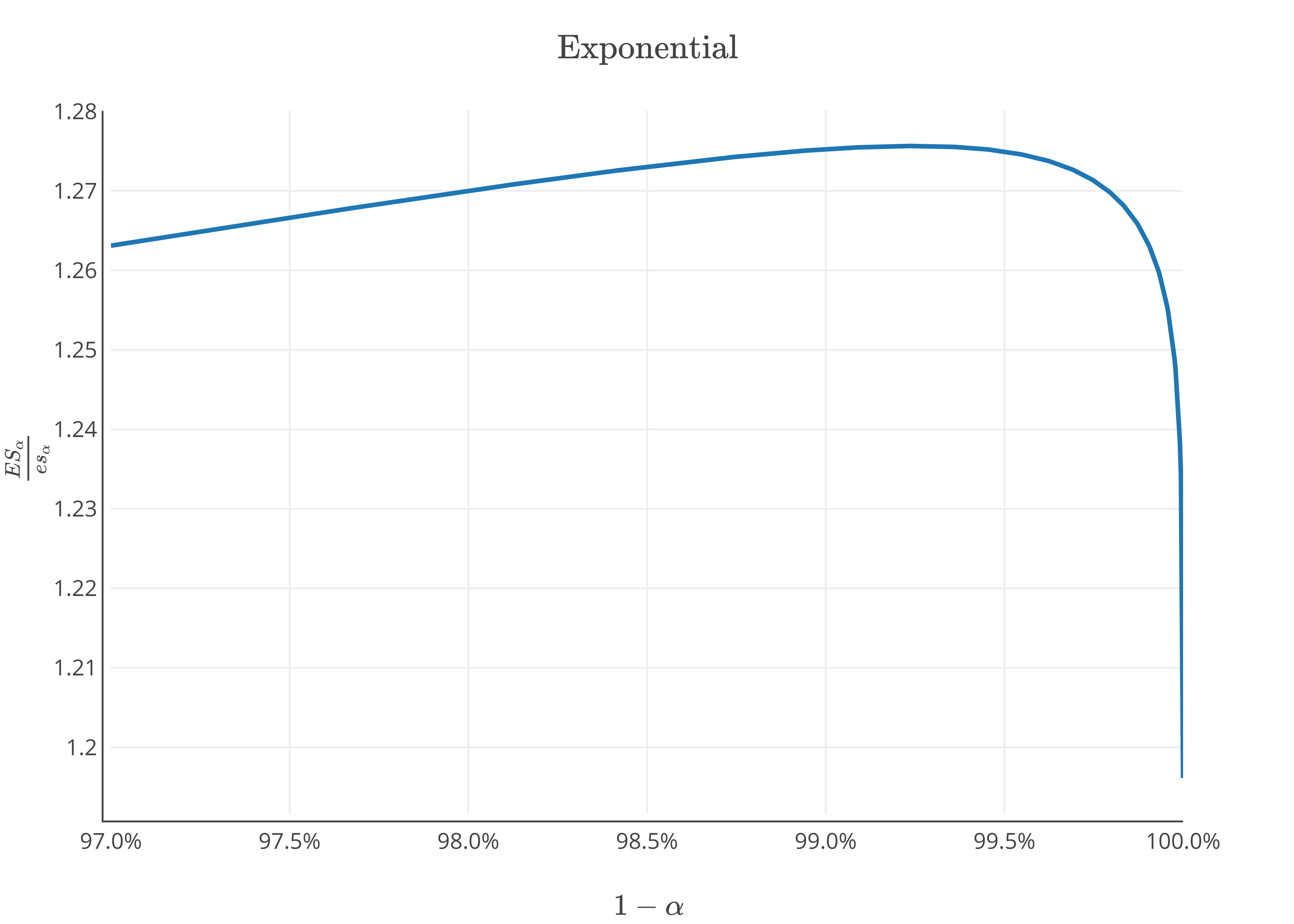}
    \end{subfigure}
    \caption{Graph of $ES_\alpha/es_\alpha$ (for Pareto and exponential) and $(\hat{x}-ES_\alpha)/(\hat{x}-es_\alpha)$ (for uniform and beta) distributions.}
    \label{fig:fig03}
\end{figure}

\section{Examples}
\begin{example}[Ber($p$)]\label{eg:bern}
    Let $L\sim Bern(p)$ for some $p$ in $(0,1)$.
    From \citep{delbaen2013,neil2015}, we get 
    \begin{equation*}
        e_\alpha(L)=\frac{(1-\alpha)p}{(1-2\alpha)p+\alpha}.
    \end{equation*}
    After integration
    \begin{equation*}
        es_\alpha =\begin{cases}
            1-\frac{\alpha}{2} & \text{if } p=1/2,\\
            -\frac{p}{1-2p}\left[1-\frac{1-p}{\alpha(1-2p)}\ln{\left(1-2\alpha+\frac{\alpha}{p}\right)}\right] & \text{if }p\neq 1/2
        \end{cases}.
    \end{equation*}
\end{example}

\begin{example}[Uniform]\label{eg:unif}
    Let $F_L(x)=x$ for $x$ in $[0,1]$. From \cite{mekonnen2019}, we have
    \begin{equation*}
        q_\alpha(L)=1-\alpha, \quad e_\alpha(L)=1-\frac{\sqrt{\alpha(1-\alpha)}-\alpha}{1-2\alpha}\quad \text{and}\quad ES_\alpha(L)=1-\frac{\alpha}{2}.
    \end{equation*}
    We also have $\beta^\ast=(\sqrt{\alpha(1-\alpha)}-\alpha)/(1-2\alpha)$ and  $tce_\alpha(L)=1-\frac{\beta^\ast}{2}$.
    After integration
    \begin{equation*}
        es_\alpha(L)=\frac{1}{2}-\frac{1}{4\alpha}\ln{\left((1-2\alpha)\sqrt{\frac{1+2\sqrt{\alpha(1-\alpha)}}{1-2\sqrt{\alpha(1-\alpha)}}}\right)}+\frac{\sqrt{\alpha(1-\alpha)}}{2\alpha}.
    \end{equation*}
    Following \citep{hua2011}, $F_L$ is attracted by Weibull type $MDA(\Psi_1)$.
    It also holds that  $\hat{x}=1$ and a simple computation yields
    \begin{equation*}
        \frac{1-ES_\alpha(L)}{1-es_\alpha(L)}= o(1) \quad \text{ and }\quad \frac{1-tce_\alpha(L)}{1-e_\alpha(L)}=\frac{1}{2}=\frac{1-ES_\alpha(L)}{1-q_\alpha(L)}.
    \end{equation*}
\end{example}

\begin{example}[Beta]
    For $a>0$, let $ F_L(x)=x^a$ with $x$ in $[0,1]$. 
    From \citep{mekonnen2019}, we have
    \begin{equation*}
        q_\alpha(L)=(1-\alpha)^{\frac{1}{a}}\quad \text{and}\quad ES_\alpha(L) =\frac{a\left(1-(1-\alpha)^{\frac{1}{a}+1}\right)}{\alpha(a+1)}.
    \end{equation*}
    From \citep{mekonnen2019}  we also have that $e_\alpha(L)=(1-\beta^\ast)^{1/a}$ and $tce_\alpha(L) =ES_{\beta^\ast}(L)$, 
    where $\beta^\ast$ solves
    \begin{equation*}
        \frac{1-\alpha}{1-2\alpha}=(1-\beta^\ast)^{\frac{1}{a}}\left(1+\frac{\beta^\ast}{a}+\frac{\alpha(a+1)}{a(1-2\alpha)}\right).
    \end{equation*}
    According to \citep{jones1994} and \citep{neil2015},  $e_\alpha(L)=q_\alpha(\tilde{L})$, where  $\tilde{L}$ is a random variable with probability density function given by
    \begin{equation*}
        g(x)=\frac{F_L(x) E[L]-E[L 1_{\{L\leq x\}}]}{\left(2(xF_L(x)-E[L 1_{\{L\leq x\}}])+E[L]-x\right)^2}, \quad  0<F_L(x)<1.
    \end{equation*}
    In the case of the Beta distribution, it follows that 
    \begin{equation*}
        g(x)=\frac{a}{a+1}\frac{x^a(1-x)}{\left(2x^{a+1}-(a+1)x+a\right)^2}, \quad 0<x<1.
    \end{equation*}
    which yields
    \begin{multline*}
        es_\alpha(L) = \frac{1}{\alpha}\int_{0}^{\alpha} q_{u}(\tilde{L})du 
        =E\left[ \tilde{L}|\tilde{L}>q_{\alpha}(\tilde{L}) \right]\\
        =\frac{a}{\alpha(a+1)}\int_{(1-\beta^\ast)^{\frac{1}{a}}}^1  \frac{x^{a+1}(1-x)}{\left(2x^{a+1}-(a+1)x+a\right)^2}dx.
    \end{multline*}
    It is known that $F_L$ belongs to the Weibull type $MDA(\Psi_1)$, see \citep{mao2012,mao2015} for instance.
    We also have $\hat{x}=1$ and as a result of Proposition \ref{prop:MDA}, it holds  
    \begin{equation*}
        \frac{1-ES_\alpha(L)}{1-es_\alpha(L)}=o(1) \quad \text{and}\quad \frac{1-tce_\alpha(L)}{1-e_\alpha(L)}\sim \frac{1-ES_\alpha(L)}{1-q_\alpha(L)}.
    \end{equation*}
\end{example}

\begin{example}[Exponential]\label{eg:exponential}
    Let $F_L(x)=1-\exp(-x)$ for $x\geq 0$.
    Then $ES_\alpha(L)=1-\ln{\alpha}$ and $e_\alpha(L) =1+\mathcal{W}\left((1-2\alpha)/(\alpha e)\right)$, where  $\mathcal{W}$ is Lambert function\footnote{$\mathcal{W}$ is a function such that $xe^x=y$ if and only if $x=\mathcal{W}(y)$.}.
    We also have
    \begin{equation*}
        g(x)=\frac{xe^{-x}}{(1-x-2e^{-x})^2},\quad x\geq 0.
    \end{equation*}
    Hence,
    \begin{equation*}
        es_\alpha(L) =\frac{1}{\alpha}\int_{1+W\left(\frac{1-2\alpha}{\alpha e}\right)}^\infty\frac{x^2e^{-x}}{(1-x-2e^{-x})^2}dx.
    \end{equation*}
    From \citep{bellini2015a}, we get $e_\alpha(L)\sim q_\alpha(L)$.
    This implies that $es_\alpha\sim ES_\alpha(L)\sim tce_\alpha(L)$.
\end{example}
\begin{example}[Konker distribution]\label{eg:konker}
    Let $F_L(x)=\frac{4+x^2+x\cdot\sqrt{x^2+4}}{2\left(x^2+4\right)}$ for $x$ in $\mathbb{R}$.
    According to \citep[Remark 3.3]{mekonnen2019} and \cite{koenker1993}, the expectile and the value at risk coincide and given by
    \begin{equation*}
        e_\alpha(L)=q_\alpha(L)=\frac{1-2\alpha}{\sqrt{\alpha(1-\alpha)}}
    \end{equation*}
    It implies that 
    \begin{equation*}
        es_\alpha(L)=ES_\alpha(L)=2\sqrt{\frac{1-\alpha}{\alpha}}.
    \end{equation*}
    In this case, $F_L$ is attracted by the Fr\'{e}chet type $MDA(\phi_2)$. 
    Clearly, it holds that 
    \begin{equation*}
       es_\alpha(L)=ES_\alpha(L)=tce_\alpha(L) \quad \text{ and }\quad \frac{tce_\alpha(L)}{e_\alpha(L)}=\frac{ES_\alpha(L)}{q_\alpha(L)}.
   \end{equation*}
   \end{example}
\begin{example}[Pareto]\label{eg:pareto}
    For $a>1$ and $x\geq 0$, let $ F_L(x)=1-(1/(x+1))^a$.
    For $a=2$,
    \begin{equation*}
        q_\alpha(L)=\frac{1}{\sqrt{\alpha}}-1, \quad ES_\alpha(L)=\frac{2}{\sqrt{\alpha}}-1\quad \text{ and }\quad e_\alpha(L)= \sqrt{\frac{1-\alpha}{\alpha}}.
    \end{equation*}
    From \citep{mekonnen2019}, we also have $\beta^\ast=\alpha/(1+2\sqrt{\alpha(1-\alpha)})$ and we get 
    \begin{equation*}
        tce_\alpha(L)=1+2e_\alpha(L) \quad \text{and} \quad es_\alpha(L)=e_\alpha(L) + \frac{\arcsin{(\sqrt{\alpha})}}{\alpha}.
    \end{equation*}
    It is known that $F_L$ is attracted by the Fr\'{e}chet type $MDA(\phi_2)$, see \citep{mao2015,maolv2012} for instance.
    A simple computation shows that  
    \begin{equation*}
       tce_\alpha(L) \sim es_\alpha(L) \sim ES_\alpha(L) \quad \text{and} \quad \frac{ES_\alpha(L)}{q_\alpha(L)}\sim 2\sim \frac{tce_\alpha(L)}{e_\alpha(L)}.
    \end{equation*}
\end{example}

\bibliographystyle{abbrvnat}
\bibliography{biblio}
\end{document}